\documentclass[11pt]{amsart} 

\usepackage[english]{babel}
\usepackage[utf8]{inputenc}
\usepackage[T1]{fontenc}
\usepackage{comment}
\usepackage{amsmath}
\usepackage{amssymb}
\usepackage{amsfonts}
\usepackage{amsthm}
\usepackage{amscd}
\usepackage{xcolor}
\usepackage{slashed}
\usepackage[hidelinks]{hyperref}
\usepackage[noabbrev,capitalise]{cleveref}
\usepackage{geometry}
\usepackage{mathtools}
\usepackage{csquotes}
\usepackage{float}
\usepackage{subcaption}
\usepackage{footnote}

\usepackage{pgfplots}
\pgfplotsset{compat=1.18}

\numberwithin{equation}{section}

\newtheorem{definition}{Definition}[section]
\newtheorem{lemma}[definition]{Lemma}

\newtheorem{theorem}[definition]{Theorem}
\newtheorem{proposition}[definition]{Proposition}

\newtheorem{remark}[definition]{Remark}

\DeclareMathOperator{\tr}{tr}
\DeclareMathOperator{\Div}{div}

\title{The Spacetime Positive Mass Theorem with Multiple Time Dimensions}

\author{Sven Hirsch}
\address{Columbia University, 2990 Broadway, New York, NY 10027, USA}
\email{sven.hirsch@columbia.edu}
\author{Alec Payne}
\address{North Carolina State University, 2311 Stinson Drive, Raleigh, NC 27607, USA}
\email{ajpayne4@ncsu.edu}
\author{Yiyue Zhang}
\address{Beijing Institute of Mathematical Sciences and Applications, Beijing, 101408, China}
\email{zhangyiyue@bimsa.cn}

\begin{document}

\begin{abstract}
We generalize the spacetime positive mass theorem to include multiple time dimensions. In particular, we show that the mass remains nonnegative in the sense that the energy $E$ is bounded from below by the trace norm of the linear momenta $J^1,\dots, J^m$. Equality in this energy inequality implies a foliation by flat submanifolds of a generalized initial data set. Moreover, under an additional umbilicity assumption, we find that the initial data set isometrically embeds into a generalized pp-wave. 
\end{abstract}
\maketitle

\section{Introduction}

Since the introduction of Kaluza--Klein theory, physicists have speculated about general relativity with extra dimensions of spacetime. Often the proposal is to include additional spatial dimensions, but a less common consideration is having more than one time dimension. In this paper, we investigate mathematical relativity with additional \emph{timelike} dimensions. This leads to interesting rigidity statements from a purely mathematical perspective.

There are many plausible objections to physical models that include extra time dimensions. For instance, introducing additional time directions could lead to particle instabilities~\cite{Dorling-Dimensionality-Time}, causality problems~\cite{ConstraintsExtraTime}, and negative probability states (``ghosts'')~\cite[Ch.\ 7]{Bars:2009msw}, even when the additional time directions are compactified~\cite{YNDURAIN199115}. This setting naturally leads to laws of motion modeled by ultrahyperbolic PDEs, which some argue lead to unpredictive physics~\cite{Tegmark-Dimensionality-Spacetime} (see also~\cite{Walter-Weinstein-determinism, WeinsteinMultipleTimeDimensions}). 

Despite the apparent difficulties, physical theories involving multiple time dimensions are still considered. For instance, both Bars' two-time physics~\cite{Bars:2009msw, Bars2001Survey} and $F$-theory~\cite{VafaF-theory} consider spacetimes with two time dimensions. The idea of multiple time dimensions has also been explored at a conceptual level in science fiction, most notably in Egan's \emph{Dichronauts} which envisions a spacetime with signature $(2,2)$~\cite{Egan2017Dichronauts}.

In this paper we study extra time dimensions in the context of the famous positive mass theorem (PMT) \cite{SchoenYau79PMT, SY1981, Witten1981, EichmairHuangLeeSchoen11, HKK, Eichmair2013JangReduction, DingSpacetimePMT}, which is one of the foundations of mathematical relativity.
Remarkably, we find that much of the analytic and geometric framework continues to apply in multiple time dimensions. 

Recall that the classical spacetime PMT for spin manifolds asserts the following.

\begin{theorem}\label{thm pmt}{\cite{Witten1981, DingSpacetimePMT}} 
Let $(M^n,g,k)$ be an asymptotically flat spin initial data set.
Assume that the dominant energy condition
\begin{align*}
    \mu \ge |J|
\end{align*}
holds, where $\mu$ denotes the energy density and $J$ the momentum density.
Then the ADM energy $E$ and linear momentum $P$ satisfy
\begin{align*}
    E \ge |P|.
\end{align*}
\end{theorem}

Here, $k$ is a symmetric $(0,2)$-tensor representing the second fundamental form of $(M^n,g)$ as a spacelike hypersurface in a Lorentzian spacetime $(\overline M^{n+1},\overline g)$.
If instead $(M,g)$ arises as an initial data set in a pseudo-Riemannian manifold $(\overline M^{n+m},\overline g)$ with $m$ timelike directions, then the second fundamental form becomes vector-valued.
Equivalently, it consists of $m$ symmetric tensors $k^1,\dots,k^m$.
In this setting, we are able to show\footnote{A related problem was studied in~\cite[Theorem 4.1]{ChenHijaziZhang2012DiracWitten}. See Remark \ref{mistakes}.}:

\begin{theorem}\label{thm main}
Let $(M^n,g,k^1,\dots,k^m)$, with $m\le n$, be an asymptotically flat spin initial data set.
Assume that the dominant energy condition
\begin{align*}
    \mu \ge \|\mathcal J\|_{\mathrm{tr}}
\end{align*}
holds, and that the $k^\alpha$, $\alpha=1,\dots,m$, pairwise commute as $g$-self-adjoint endomorphisms.
Then the energy--momentum satisfies
\begin{align*}
    E \ge \|\mathcal P\|_{\mathrm{tr}}.
\end{align*}
\end{theorem}
Here, $\mathcal J=(J^1,\dots,J^m)$, $\mathcal P=(P^1,\dots,P^m)$ are $m\times n$ matrices, consisting of rows of momentum densities $J^{\alpha}$ and momentum vectors $P^{\alpha}$, and $\|\cdot\|_{\mathrm{tr}}$ denotes the pointwise trace norm. Our definitions of these objects, as well as those of initial data sets $(M^n, g, k^1, \dots, k^m)$, the energy $E$, and the energy density $\mu$, can be found in Sections~\ref{sec IDS} and~\ref{sec trace norm}. These reduce to their classical counterparts for $m=1$, and when the extra $k^\alpha$ are small in $C^1$, $\mu$ and $\|\mathcal{J}\|_{\mathrm{tr}}$ approximate the usual energy and momentum densities. We remark that we may remove the assumption that $k^\alpha$ commute by redefining our dominant energy condition to include the the anticommutator of $k^\alpha$ (see Remarks \ref{remark geometric meaning of commuting kalpha} and \ref{commutativity}). 

We emphasize that our definitions do not come from a generalization of the Einstein field equations to multiple time dimensions. Rather, we suitably generalize all quantities from the $m=1$ initial data set formulation to arrive at natural mathematical results. This is further justified on mathematical grounds by our rigidity results in Theorems \ref{thm rigidity 1} and \ref{thm rigidity 2}.

In the case of a single time dimension, the PMT admits a strong rigidity statement: initial data sets saturating the inequality in Theorem~\ref{thm pmt} must embed into Minkowski space or into a pp-wave spacetime \cite{HirschJangZhang24, HirschHuang25, BeigChrusciel, ChruscielMaerten05, HuangLee20, Huang-Lee:2024, HZ23}.
Such pp-wave spacetimes model gravitational radiation and are geometrically characterized by the existence of a codimension $1$ foliation of flat hypersurfaces in the initial data set $(M,g,k)$.

With multiple timelike directions, the rigidity problem becomes substantially more subtle.
This phenomenon already appears in the positive mass theorem with cosmological constant, which effectively behaves as the presence of an additional timelike dimension \cite{hirsch2025causal}.
In that context, equality is attained, among other examples, by Siklos wave spacetimes, whose initial data sets admit a codimension $2$ foliation by flat submanifolds.
Also, see Cecchini--Lesourd--Zeidler's proof of the spacetime PMT with shields which makes use of two time dimensions indirectly with $k^1=k$ and $k^2=fg$~\cite{CecchiniLesourdZeidler2024PositiveMass}.

Our first rigidity result is the following.

\begin{theorem}\label{thm rigidity 1}
Suppose equality holds in Theorem~\ref{thm main}, that is,
\begin{align*}
    E=\|\mathcal P\|_{\mathrm{tr}}.
\end{align*}
Then, $M$ is foliated by flat submanifolds $\Sigma$ of codimension $m$.
\end{theorem}

Under an additional umbilicity assumption, this foliation can be promoted to an embedding result.

\begin{theorem}\label{thm rigidity 2}
Assume, in addition to the hypotheses of Theorem~\ref{thm rigidity 1}, that
\[
k^\alpha = f^\alpha g \qquad \text{for } \alpha=1,\dots,m,
\]
for some functions $f^\alpha$.
Then $(M,g,k^1,\dots,k^m)$ admits an isometric embedding into a generalized pp-wave $(\overline M^{n+m},\overline g)$ whose second fundamental form is given by $(k^1,\dots,k^m)$.
Moreover, the normal bundle is trivial.
\end{theorem}

Here, a pseudo-Riemannian spin manifold $(\overline M^{n+m},\overline g)$ is called a \emph{generalized pp-wave} if it admits a nontrivial parallel null spinor $\overline\psi\in\mathcal S(\overline M^{n+m})$. 
This reduces to a standard pp-wave in case $m=1$.

 Theorem \ref{thm main} is proven via a Witten-type divergence formula. 
In the case of equality, this yields a spinor $\psi\in\overline{\mathcal S}(M)$ solving
\begin{align}\label{psi overdetermined intro}
    \nabla_i\psi=-\frac12\sum_{\alpha=1}^m k^\alpha_{ij}e_j\eta_\alpha\psi.
\end{align}
The geometric properties of a manifold strongly depend on the causal character of its underlying spinor.
A consequence of our techniques, which may be of independent interest, is the following (see Theorem \ref{N=X} and Section \ref{section spinorial symmetries}).

\begin{theorem}\label{main technical}
Let $\psi$ be as in \eqref{psi overdetermined intro}, and let $\mathcal{M} = \{1, \dots, m\}$.
Then either $\psi$ is everywhere null or $\psi$ is everywhere timelike.
In the latter case, the spinor
\begin{align*}
    \varphi=\eta_\mathcal{M} \sum_{|I|=|\Gamma|} (-1)^{|\Gamma|}\langle e_I \eta_\Gamma\psi,\psi\rangle e_I \eta_\Gamma\psi
\end{align*}
also solves 
\begin{align*}
    \nabla_i\varphi+\frac12\sum_{\alpha=1}^mk^\alpha_{ij}e_j\eta_\alpha\varphi=0.
\end{align*}
\end{theorem}
Given a spinor $\psi$, let $N=|\psi|^2$ and $X^\alpha_i(\psi)=\langle e_i\eta_\alpha\psi,\psi\rangle$ (see Definition \ref{definition N and X}).
We say $\psi$ is null if $N=|X^\alpha|$ for all $\alpha$, and say $\psi$ is timelike otherwise.

To prove Theorem \ref{thm rigidity 1}, we follow the strategy of \cite{HZ24, hirsch2025causal}. We show that there always exists a null spinor $\psi$ saturating the inequality of Theorem \ref{thm main} and then use Theorem \ref{main technical} to show that the null condition is preserved throughout $M$.
To prove the latter in the classical setting ($m=1$) in \cite{HZ24}, we consider the function $N$ and the vector field $X$ as above and show that $\nabla_i(N^2-|X|^2)=0$.
With $m$ time dimensions, we need to extract $2^m$ different differential forms out of the spinor and study systems of differential equations.
More precisely, we are able to show that
\begin{align*}
    |\nabla_i F|\le CF
\end{align*}
for some constant $C>0$, where
\begin{align*}
    F=\sum_{|I|=|\Gamma|}|N^{|I|-1}\omega_I^\Gamma-\widetilde{X}^\Gamma_I|^2
\end{align*}
for multi-indices $I,\Gamma$.
Here,
\begin{align*}
 X^{\sigma(\Gamma)}_{I}=X^{\sigma_1}_{i_1}\cdots X^{\sigma_{\ell}}_{i_{\ell}},\qquad    \widetilde{X}^{\Gamma}_I=(-1)^{\frac{\ell(\ell-1)}{2}}\sum_{\sigma\in S_\Gamma} X_{I}^{\sigma(\Gamma)}e_{I}\eta_{\sigma(\Gamma)},\qquad \omega^\Gamma_I=\langle e_I\eta_\Gamma \psi,\psi\rangle e_I\eta_\Gamma,
\end{align*}
where  $S_\Gamma$  is the permutation group of $\Gamma$.

Finally, to prove Theorem \ref{thm rigidity 2} we use multiple Killing developments.

\noindent \textbf{Acknowledgments:}  
SH and AP would like to thank the Banff International Research Station and the Chennai Mathematical Institute for supporting the Connections among Spin Geometry, Minimal Surfaces, and Relativity workshop in January 2026, where part of this work was carried out. They are also grateful to the Lonavala Geometry Festival for providing ideal working conditions in which part of this work was completed. YZ was partially supported by NSFC grant No.\ 12501070 and the startup fund from BIMSA.
The authors would also like to thank Hubert Bray, Marcus Khuri and Rudolf Zeidler for helpful discussions and their interest in this work.

\section{Preliminaries}

\subsection{Asymptotically flat initial data sets}\label{sec IDS}

\begin{definition}
We say that $(M^n,g,k^1,\dots,k^m)$ is an \emph{asymptotically flat initial data set with decay rate}
$q>\frac{n-2}{2}$ if, for each $\alpha=1,\dots,m$, the pair $(g,k^\alpha)$ is asymptotically flat with decay rate $q$ in the standard sense.
\end{definition}

\begin{definition}
Given an initial data set $(M^n,g,k^1,\dots,k^m)$, we define the \emph{energy density} $\mu$ and the
\emph{momentum densities} $J^\alpha$, $\alpha=1,\dots,m$, by
\begin{align*}
\mu
&= \frac{1}{2}\left[R_g + \sum_{\alpha=1}^m\bigl(\tr_g(k^\alpha)^2 - |k^\alpha|^2\bigr)\right],\\
J^\alpha
&= \Div_g\!\left(k^\alpha - \tr_g(k^\alpha)\, g\right),
\end{align*}
where $R_g$ denotes the scalar curvature of $g$.
\end{definition}

\begin{definition}
Let $(M^n,g,k^1,\dots,k^m)$ be an asymptotically flat initial data set.
The \emph{energy} $E$ and the \emph{momentum vectors} $P^\alpha$, $\alpha=1,\dots,m$, are defined by
\begin{align*}
E
&=\frac{1}{2(n-1)\omega_{n-1}}
\lim_{r\to\infty}
\int_{S_r}\bigl(g_{ij,i}-g_{ii,j}\bigr)\,\nu^j,\\
P^\alpha_i
&=\frac{1}{(n-1)\omega_{n-1}}
\lim_{r\to\infty}
\int_{S_r}\bigl(k^\alpha_{ij}-\tr_g(k^\alpha)g_{ij}\bigr)\,\nu^j,
\end{align*}
where $S_r$ denotes the coordinate sphere of radius $r$ in the asymptotically flat end
and $\nu$ its outward unit normal.
\end{definition}

As in the classical case $m=1$, the above decay assumptions ensure that $E$ and each $P^\alpha$ are well-defined and independent of the choice of asymptotically flat coordinates. For the initial data set $(M^n, g, k^1, 0, \dots, 0)$, the energy density $\mu$, energy $E$, momentum density $J^1$, and momentum vector $P^1$ all coincide with their classical counterparts for the initial data set $(M^n, g, k^1)$. 

\begin{remark}\label{remark on definitions} Our definitions in this section do not come from an analysis of the Einstein field equations in multiple time dimensions. Instead, we have generalized the classical definitions from the initial data set formulation in a way that yields the most natural mathematical results.
\end{remark}

\subsection{Spin geometry}

Let $m\in\mathbb N$ and let $(M^n,g,k^1,\dots,k^m)$ be a spin initial data set with symmetric
$(0,2)$-tensors $k^1,\dots,k^m$.
We denote by $\mathcal S(M^n)$ the spinor bundle of $(M^n,g)$ and set
\[
\overline{\mathcal S}(M^n)=\mathcal S(M^n)^{2^m}.
\]

Let $\{e_1,\dots,e_n,\eta_1,\dots,\eta_m\}$ be an orthonormal basis of $\mathbb{R}^{n,m}$. We construct a Clifford action of $\mathbb R^{n,m}
=\operatorname{span}\{e_1,\dots,e_n,\eta_1,\dots,\eta_m\}$
on $\overline{\mathcal S}(M^n)$.
Proceeding inductively, for
\[
(\psi_1,\psi_2)\in \mathcal S(M^n)^{2^\alpha}
=\mathcal S(M^n)^{2^{\alpha-1}}\oplus \mathcal S(M^n)^{2^{\alpha-1}},
\]
we define
\begin{align*}
e_i(\psi_1,\psi_2)
&=(e_i\psi_1,-e_i\psi_2),
&&1\le i\le n,\\
\eta_\beta(\psi_1,\psi_2)
&=(\eta_\beta\psi_1,-\eta_\beta\psi_2),
&&1\le \beta\le \alpha-1,\\
\eta_\alpha(\psi_1,\psi_2)
&=(\psi_2,\psi_1).
\end{align*}
A direct computation shows that these operators satisfy the Clifford relations
\begin{align*}
e_ie_j+e_je_i&=-2\delta_{ij},\\
\eta_\alpha\eta_\beta+\eta_\beta\eta_\alpha&=2\delta_{\alpha\beta},\\
e_i\eta_\alpha+\eta_\alpha e_i&=0,
\end{align*}
corresponding to signature $(n,m)$.
The bundle $\overline{\mathcal S}$ can also be interpreted as $\mathcal S$ twisted with a flat bundle as in \cite{WangXieYu2024CubeInequality}.

The bundle $\overline{\mathcal S}(M^n)$ carries a natural Hermitian inner product,
defined inductively by
\[
\langle(\psi_1,\psi_2),(\phi_1,\phi_2)\rangle
=\langle\psi_1,\phi_1\rangle+\langle\psi_2,\phi_2\rangle.
\]
With respect to this inner product, Clifford multiplication satisfies
\begin{align*}
\langle e_i\psi,\phi\rangle&=-\langle\psi,e_i\phi\rangle,\\
\langle \eta_\alpha\psi,\phi\rangle&=\langle\psi,\eta_\alpha\phi\rangle,
\end{align*}
for all $\psi,\phi\in\overline{\mathcal S}(M^n)$.

\begin{definition}\label{definition N and X}
Given a spin initial data set $(M^n,g,k^1,\dots,k^m)$ and a spinor
$\psi\in\overline{\mathcal S}(M^n)$, we define the associated scalar field
\begin{align*}
N(\psi)&=|\psi|^2.
\end{align*}
We also define the associated vector fields $X^{\alpha}(\psi)$ on $M$ by
\begin{align*}
\langle X^{\alpha}(\psi), Y\rangle &=\langle Y \eta_\alpha\psi,\psi\rangle,
\end{align*}
where $Y$ is an arbitrary vector field on $M$.  In particular, for a coordinate vector $e_i$, 
\begin{align*}
    X^{\alpha}_i(\psi) = \langle e_i \eta_{\alpha} \psi, \psi\rangle.
\end{align*}
For brevity, we denote $N(\psi)$ and $X^\alpha(\psi)$ by $N$ and $X^\alpha$.
We further introduce the modified spin connection and Dirac operator
\begin{align*}
\overline{\nabla}_i\psi
&=\nabla_i\psi+\frac12\sum_{\alpha=1}^m k^\alpha_{ij}e_j\eta_\alpha\psi,\\
\overline{\slashed D}\psi
&=e_i\overline{\nabla}_i\psi.
\end{align*}
\end{definition}

\begin{definition}
A spinor $\psi\in\overline{\mathcal S}(M^n)$ is called a \emph{null spinor} if,
for every choice of unit timelike vector $\eta^\alpha$,
the associated pair $(N(\psi),X^\alpha(\psi))\in\mathbb R^{n,1}$ is null, i.e.
\[
N(\psi)=|X^\alpha(\psi)|.
\]
\end{definition}

We will also use multi-index notation.
Let $I=\{i_1,\dots,i_\ell\}\subset\{1,\dots,n\}$ and
$\Gamma=\{\alpha_1,\dots,\alpha_k\}\subset\mathcal M:=\{1,\dots,m\}$,
with elements listed in increasing order.
We write
\[
e_I=e_{i_1}\cdots e_{i_\ell},\qquad
\eta_\Gamma=\eta_{\alpha_1}\cdots \eta_{\alpha_k}.
\]
Further Clifford algebra-valued forms built from these multi-indices will be introduced
as needed in later sections.

\subsection{Trace Norm and Dominant Energy Condition}\label{sec trace norm}

Let $\mathcal P$ be an $m\times n$ matrix, and let its singular values be given by
\[
\sigma_1(\mathcal P)\ge \sigma_2(\mathcal P)\ge \dots \ge \sigma_{\min(m,n)}(\mathcal P)\ge 0.
\]
The \emph{trace norm} of $P$ is the sum of the singular values of $P$, i.e.\
\[
\|\mathcal P\|_{\mathrm{tr}}=\sum_{i=1}^{\min(m,n)}\sigma_i(\mathcal P).
\]
This norm arises naturally from the singular value decomposition:
there exist an $m\times m$ orthogonal matrix $U$, an $n\times n$ orthogonal matrix $V$,
and an $m\times n$ rectangular diagonal matrix $A$ such that
\[
\mathcal P=U^TAV,
\]
where the diagonal entries of $A$ are precisely the singular values $\sigma_i(\mathcal P)$.
The trace norm is therefore the sum of the diagonal entries of $A$.

Using the trace norm, we can formulate our dominant energy condition with extra time dimensions.

\begin{definition}
    We say that an initial data set $(M^n,g,k^1,\dots,k^m)$ satisfies the dominant energy condition if 
    \begin{align*}
        \mu \geq \|\mathcal J\|_{\mathrm{tr}},
    \end{align*}
where $\|\mathcal J\|_{\mathrm{tr}}=\|(J^1,\dots,J^m)\|_{\mathrm{tr}}$ denotes the pointwise trace norm of the $m \times n$ matrix $\mathcal{J}$ with rows $J^1, \dots, J^m$.
\end{definition}

In $F$-theory, the additional dimensions are typically assumed to be small. From this perspective, we note that our dominant energy condition approximates the classical one when the extra $k^{\alpha}$ are small in $C^1$.

\section{The positive mass inequality}

\begin{lemma}[Divergence identity]\label{lemma divergence identity}
Let $(M^n,g,k^1,\dots,k^m)$ be a spin initial data set, and assume that the tensors
$k^\alpha$, $\alpha=1,\dots,m$, pairwise commute as $g$-self-adjoint endomorphisms.
Then, for any spinor $\psi\in\overline{\mathcal S}(M^n)$,
\begin{align*}
\nabla_i\Bigl(
&\langle e_i\slashed D\psi+\nabla_i\psi,\psi\rangle
+\sum_{\alpha=1}^m (k^\alpha_{ij}-\tr_g(k^\alpha)g_{ij})
   \langle \eta_\alpha e_j\psi,\psi\rangle
\Bigr) \\
={}&
\Bigl|\nabla_i\psi+\frac12\sum_{\alpha=1}^m k^\alpha_{ij}e_j\eta_\alpha\psi\Bigr|^2
-\Bigl|\slashed D\psi-\frac12\sum_{\alpha=1}^m \tr_g(k^\alpha)\eta_\alpha\psi\Bigr|^2 \\
&\quad
+\frac14\Bigl(
R_g+\sum_{\alpha=1}^m\bigl(\tr_g(k^\alpha)^2-|k^\alpha|^2\bigr)
\Bigr)|\psi|^2
+\frac12\sum_{\alpha=1}^m
\bigl\langle \Div_g(k^\alpha-\tr_g(k^\alpha)g)\,\eta_\alpha\psi,\psi\bigr\rangle .
\end{align*}
\end{lemma}

This is the only point in the paper where we use that $k^\alpha$ commute.

\begin{remark}
In case $k^\alpha=f^\alpha g$ this also has applications in Riemannian geometry.
For instance, for an appropriate choice of functions $f^\alpha$,  this recovers Wang--Xie--Yu's proof~\cite{WangXieYu2024CubeInequality} of Gromov's cube inequality~\cite[Section 3.8]{Gromov2019FourLectures}.
See also~\cite[Theorem 1.4]{HirschKazarasKhuriZhang2023SpectralToricBand}.
\end{remark}

\begin{proof}
We first expand
\begin{align}\label{eq gradient squared}
\begin{split}
\Bigl|\nabla_i\psi+\frac12\sum_{\alpha=1}^m k^\alpha_{ij}e_j\eta_\alpha\psi\Bigr|^2
={}&
|\nabla\psi|^2
+\frac14\sum_{\alpha=1}^m |k^\alpha|^2|\psi|^2
+\sum_{\alpha=1}^m\langle\nabla_i\psi,k^\alpha_{ij}e_j\eta_\alpha\psi\rangle \\
&\quad
+\frac14\sum_{\substack{\alpha,\beta=1\\ \alpha\neq\beta}}^m
k^\alpha_{ij}k^\beta_{il}
\operatorname{Re}\langle e_j\eta_\alpha\psi,e_l\eta_\beta\psi\rangle .
\end{split}
\end{align}
Moreover, we obtain
    \begin{align}\label{eq anticommutator}
    \begin{split}
        \sum_{\substack{1 \leq \alpha,\beta \leq m\\ \alpha\neq\beta}} 
k^\alpha_{ij}k^\beta_{il}
&\operatorname{Re}\langle e_j\eta_\alpha\psi,e_l\eta_\beta\psi\rangle
\\=&
\sum_{\substack{1 \leq \alpha < \beta \leq m}}
k^\alpha_{ij}k^\beta_{il}
\operatorname{Re}\langle e_j\eta_\alpha\psi,e_l\eta_\beta\psi\rangle 
+
\sum_{\substack{1 \leq \beta < \alpha \leq m}}
k^\alpha_{ij}k^\beta_{il}
\operatorname{Re}\langle e_j\eta_\alpha\psi,e_l\eta_\beta\psi\rangle 
\\=&
\sum_{\substack{1 \leq \alpha < \beta \leq m}}
k^\alpha_{ij}k^\beta_{il}
\operatorname{Re}\langle e_j\eta_\alpha\psi,e_l\eta_\beta\psi\rangle 
+
\sum_{\substack{1 \leq \beta < \alpha \leq m}}
k^\beta_{ij}k^\alpha_{il}
\operatorname{Re}\langle e_j\eta_\beta\psi,e_l\eta_\alpha\psi\rangle 
\\=&
\sum_{\substack{1 \leq \alpha < \beta \leq m}}
\left(k^\alpha_{ij}k^\beta_{il}-k^\beta_{ij}k^\alpha_{il}\right)
\operatorname{Re}\langle e_j\eta_\alpha\psi,e_l\eta_\beta\psi\rangle.
\end{split}
    \end{align}
    Since $k^\alpha$ commute as $g$-self-adjoint endomorphisms, this term vanishes.
    Consequently, \eqref{eq gradient squared} becomes
    \begin{align*}
\Bigl|\nabla_i\psi+\frac12\sum_{\alpha=1}^m k^\alpha_{ij}e_j\eta_\alpha\psi\Bigr|^2
&=
|\nabla\psi|^2
+\frac14\sum_{\alpha=1}^m |k^\alpha|^2|\psi|^2
+\sum_{\alpha=1}^m\langle\nabla_i\psi,k^\alpha_{ij}e_j\eta_\alpha\psi\rangle.
\end{align*}

Similarly,
\begin{align*}
\Bigl|\slashed D\psi-\frac12\sum_{\alpha=1}^m \tr_g(k^\alpha)\eta_\alpha\psi\Bigr|^2
={}&
|\slashed D\psi|^2
+\frac14\sum_{\alpha=1}^m (\tr_g(k^\alpha))^2|\psi|^2
-\sum_{\alpha=1}^m\langle \slashed D\psi,\tr_g(k^\alpha)\eta_\alpha\psi\rangle \\
&\quad
+\frac14\sum_{\substack{1 \leq \alpha,\beta \leq m\\ \alpha\neq\beta}}
\tr_g(k^\alpha)\tr_g(k^\beta)
\langle\eta_\alpha\psi,\eta_\beta\psi\rangle .
\end{align*}
The final term vanishes since
$\langle\eta_\alpha\psi,\eta_\beta\psi\rangle$ is antisymmetric for $\alpha\neq\beta$ whereas $\tr_g(k^\alpha)\tr_g(k^\beta)$ is symmetric.

Using the classical Lichnerowicz formula
\[
|\nabla\psi|^2+\frac14 R_g|\psi|^2-|\slashed D\psi|^2
=\nabla_i\bigl(\langle\psi,\nabla_i\psi\rangle
-\langle e_i\psi,\slashed D\psi\rangle\bigr),
\]
and combining the above identities yields the stated divergence formula.
\end{proof}

\begin{remark}\label{remark geometric meaning of commuting kalpha}
We give some geometric meaning to our assumption that the tensors $k^{\alpha}$ pairwise commute. Suppose that $(M, g, k^1, \dots, k^m)$ admits an isometric embedding into $(\overline{M}^{n+m}, \overline{g})$ with second fundamental form $(k^1, \dots, k^m)$, as in Theorem \ref{thm rigidity 2}. If the $k^{\alpha}$ commute, then the Ricci equation implies 
    \begin{align*}
        R^{\perp}_{ij \alpha \beta} = \overline{R}_{ij \alpha \beta}.
    \end{align*}
    Here, $\overline{R}$ is the curvature of $\overline{M}$, and $R^{\perp}$ is the curvature of the connection induced on the normal bundle of the embedding of $M$ in $\overline{M}$. Also, $\alpha, \beta$ (resp.\ $k,l$) are normal (resp.\ tangential) directions to the embedding of $M$.
\end{remark}

\begin{remark}\label{commutativity}
    As mentioned in the introduction, we could have defined the dominant energy condition to include the anticommutator term appearing in \eqref{eq anticommutator} which would remove our assumption that $k^{\alpha}$ commute.
      However, we opted against this since this appears to be physically unnatural.
        More precisely, the DEC becomes in this case
    \begin{align}\label{DEC new}
        \mu-\|\mathcal J\|_{\mathrm{tr}}\ge \left\lfloor\frac n2\right\rfloor \sum_{\alpha<\beta}\|[k^\alpha,k^\beta]\|_\ast,
    \end{align}
    where $\|\cdot\|_\ast$ denotes the comass norm of a two-form and
    \begin{align*}
      [k^\alpha,k^\beta]=  \frac12\left(k^\alpha_{ij}k^\beta_{il}-k^\beta_{ij}k^\alpha_{il}\right).
    \end{align*}
    In particular, in case this generalized DEC \eqref{DEC new} holds and $\slashed D\psi=\frac12\sum_{\alpha=1}^m \tr_g(k^\alpha)\eta_\alpha\psi$, the RHS of the divergence identity of Lemma \ref{lemma divergence identity} is non-negative, even in case $k^\alpha$ do not commute.
    Furthermore, as pointed out above, Lemma \ref{lemma divergence identity} is the only point of the paper which uses that $k^\alpha $ commute.
    Consequently, Theorem \ref{thm main}, Theorem \ref{thm rigidity 1} and Theorem \ref{thm rigidity 2} still hold in this setting of this generalized DEC \eqref{DEC new}.
    \end{remark}

\begin{lemma}\label{eq mu J positivity}
    Suppose that $ \mu \ge \|\mathcal J\|_{\mathrm{tr}}$. Then
    \begin{align*}
\mu N(\psi)
+\sum_{\alpha=1}^m\langle J^\alpha,X^\alpha(\psi)\rangle
\ge 0 .
\end{align*}
\end{lemma}

\begin{proof}
  Expanding the second term, we have
\begin{align*}
    \sum_{\alpha=1}^m\langle J^\alpha,X^\alpha(\psi)\rangle=\langle J^\alpha \eta_\alpha\psi,\psi \rangle. 
\end{align*}
Let $\{\widehat e_i,\widehat\eta_\alpha\}$ be an orthonormal frame such that
$e_i=V_{ij}\widehat{e}_j$ and $\eta_\alpha=U_{\alpha\beta}\widehat{\eta}_\beta$,
with $U\in O(m)$ and $V\in O(n)$.
Then,
\begin{align*}
J^\alpha\eta_\alpha
&=J^\alpha_i e_i\eta_\alpha
=J^\alpha_i V_{ij}\widehat{e}_j U_{\alpha\beta}\widehat{\eta}_\beta
=( U^T\mathcal{J}V)_{\beta j}\widehat{e}_j\widehat{\eta}_\beta.
\end{align*}
Choosing  $U$ and $V$ from the singular value decomposition of $\mathcal{J}$,
the matrix $A:= U^T\mathcal{J} V$ is rectangular diagonal with
$A_{ss}=\sigma_s(\mathcal{J})$. Hence,
\begin{align}\label{equation thm main result lower bound on second term}
\langle J^\alpha\eta_\alpha\psi,\psi\rangle
&=\sum_{s=1}^m \sigma_s(\mathcal{J})
\langle \widehat{e}_s\widehat{\eta}_s\psi,\psi\rangle
\ge -\sum_{s=1}^m \sigma_s(\mathcal{J})|\psi|^2
=-\|\mathcal{J}\|_{\mathrm{tr}}N(\psi).
\end{align}
Thus, the result follows.
\end{proof}

\begin{theorem}\label{thm integral identity}
Let $(M^n,g,k^1,\dots,k^m)$ be an asymptotically flat initial data set satisfying the dominant
energy condition
\[
\mu\ge \|\mathcal J\|_{\mathrm{tr}} .
\]
Let $\psi^\infty$ be a constant spinor on $\overline{\mathcal S}(\mathbb R^n)$, which we identify with a constant spinor in $\overline{\mathcal S}(M)$.
Then there exists a spinor $\psi \in \overline{\mathcal S}(M)$ with
$\psi-\psi^\infty\in C^{2,\alpha}_{-q}$ solving
$\overline{\slashed D}\psi=0$, such that
\begin{align*}
EN(\psi^\infty)
+\sum_{\alpha=1}^m\langle P^\alpha,X^\alpha(\psi^\infty)\rangle
=\frac{2}{(n-1)\omega_{n-1}}
\int_{M^n}\Bigl(
|\overline{\nabla}\psi|^2
+\frac12\mu N(\psi)
+\frac12\sum_{\alpha=1}^m\langle J^\alpha,X^\alpha(\psi)\rangle
\Bigr)\ge0.
\end{align*}
\end{theorem}

\begin{proof}
Integrating the divergence identity in Lemma~\ref{lemma divergence identity} over $M^n$
and applying the divergence theorem yields boundary integrals over $S_r$.
Using asymptotic flatness, one computes, as in \cite[Proposition 8.24]{LeeGeometricRelativity},
\begin{align*}
\lim_{r\to\infty}\int_{S_r}
\langle \nu\slashed D\psi+\nabla_\nu\psi,\psi\rangle
&=\frac12 (n-1)\omega_{n-1} \,E\,N(\psi^\infty),\\
\lim_{r\to\infty}\int_{S_r}
(k^\alpha_{\nu j}-\tr_g(k^\alpha)g_{\nu j})
\langle\eta_\alpha e_j\psi,\psi\rangle
&=\frac12 (n-1)\omega_{n-1}
\langle P^\alpha,X^\alpha(\psi^\infty)\rangle ,
\end{align*}
where $\nu$ denotes the outward unit normal to $S_r$.
This proves the stated integral identity.

The existence theory for solutions of the equation $\overline{\slashed D}\psi=0$ essentially follows \cite{LeeGeometricRelativity} where it is proven for $m=1$.
It is well-known \cite[Proposition 5.16]{LeeGeometricRelativity} that $\slashed D:W^{1,2}_{-q}(\mathcal S(M))\to L^2_{-q-1}(\mathcal S(M))$ is an isomorphism.
Applying this result $2^m$ times, we obtain that $\slashed D:W^{1,2}_{-q}(\overline{\mathcal S}(M))\to L^2_{-q-1}(\overline{\mathcal S}(M))$ is also an isomorphism.
Since we can continuously deform $\overline{\slashed D}$ to $\slashed D$ in the strong operator topology, their Fredholm indices coincide and both vanish, cf. \cite[Proposition 8.27]{LeeGeometricRelativity}.
If $\overline {\slashed D}\psi=0$, we obtain $\overline \nabla\psi=0$ from the integral formula.
Therefore, $|\nabla |\psi||\le C|\psi|$ where $C=C(k^1,\dots,k^m)$.
Integrating this ODE and using $\psi\to0$ at $\infty$, we must have $\psi\equiv0$ which shows injectivity.
Consequently, 
 $$\overline{\slashed D}:W^{1,2}_{-q}(\overline{\mathcal S}(M))\to L^2_{-q-1}(\overline{\mathcal S}(M))$$
 is an isomorphism.
Next, fix a constant spinor $\psi^\infty \in \overline{\mathcal S}(\mathbb{R}^n)$ which we extend to $M$.
Then $\overline{\slashed D}\psi^\infty\in L^2_{-q-1}$ and we can find another spinor $\sigma\in W^{1,2}_{-q}(\overline{\mathcal S}(M))$ solving 
$$
\overline {\slashed D}\sigma=\overline{\slashed D}\psi^\infty.
$$
Define $\psi=\psi^\infty-\sigma$ which solves $\overline{\slashed D}\psi=0$ and $\psi\to\psi^\infty$ at infinity.
 
Finally, we proceed as in \cite[Lemma 3.2]{HirschJangZhang24} to obtain higher order estimates for $\sigma $ and $\psi$.
\end{proof}

\begin{proof}[Proof of Theorem~\ref{thm main}]
For any choice of constant spinor $\psi^\infty$, Lemma \ref{lemma divergence identity} and Theorem \ref{thm integral identity} imply
\begin{align}\label{equation integral identity}
EN(\psi^\infty)
+\sum_{\alpha=1}^m\langle P^\alpha,X^\alpha(\psi^\infty)\rangle
\ge 0 .
\end{align}

Let $\{e^\infty_i\}$ be an orthonormal basis in $\mathbb{R}^n$. 
Similar to the proof of Lemma \ref{eq mu J positivity}, let ${\widetilde e^\infty_i,\widetilde\eta_\alpha}$ be an orthonormal frame with $e^\infty_i = \mathcal{V}_{ij} \widetilde e_j$ and $\eta_\alpha = \mathcal{U}_{\alpha\beta} \widetilde\eta_\beta$, where $\mathcal{U}$ and $\mathcal{V}$ are the orthogonal matrices from the singular value decomposition of $\mathcal{P}$ so that $\mathcal{U}^T\mathcal{P}\mathcal{V}$ is rectangular diagonal. Choose $\psi^\infty$ satisfying $\widetilde{e}^\infty_\alpha \widetilde{\eta}_\alpha \psi^\infty = -\psi^\infty$ for $1\le \alpha\le m$. Then, arguing similarly to the proof of Lemma \ref{eq mu J positivity}, we have that
$$\sum_{\alpha=1}^m \langle P^{\alpha}, X^{\alpha}(\psi^{\infty})\rangle = - \|\mathcal P\|_{\mathrm{tr}}N(\psi^{\infty}),$$
and so
\begin{align*}
    E N(\psi^{\infty}) - \|\mathcal{P}\|_{\mathrm{tr}}N(\psi^\infty)\geq 0.
\end{align*}
This completes the proof.
\end{proof}

\begin{remark}\label{mistakes}
In \cite[Theorem 4.1]{ChenHijaziZhang2012DiracWitten} a related statement to Theorem \ref{thm main} is discussed, though they do not cover our mass rigidity statements of Theorems \ref{thm rigidity 1} and \ref{thm rigidity 2}.
There, the authors additionally assume that $(M,g,k^1,\dots,k^m)$ arises as submanifold of a pseudo-Riemannian manifold and claim the weaker inequality $$E\ge\sqrt{|P^1|^2+\dots+|P^m|^2}.$$
On the other hand, they assume a different dominant energy condition \cite[equation 5]{ChenHijaziZhang2012DiracWitten}.
However, there are several issues with this notion and their underlying algebra, which are vital for the result. 
First, we could not verify the computation of \cite[Theorem 2.5]{ChenHijaziZhang2012DiracWitten}.
Second, the usage of the Euclidean norm in their DEC is insufficient to conclude non-negativity of the integrand in the mass formula~\cite[Proof of Thm.\ 4.1]{ChenHijaziZhang2012DiracWitten}.
\end{remark}

\section{Case of equality}   

Our rigidity proof follows the argument for $m=2$ in \cite{hirsch2025causal} in the null case.
More precisely, our arguments below generalize \cite[Theorem 5.4]{hirsch2025causal}, \cite[Proposition 5.7]{hirsch2025causal}, \cite[Proposition 5.9]{hirsch2025causal}, \cite[Theorem 6.1]{hirsch2025causal} and \cite[Section 7]{hirsch2025causal}.
The proof strategy remains the same, although the notation becomes significantly more involved.

First, we show that we can construct several null spinors solving the overdetermined equation
    \begin{equation*}
    \nabla_i\psi+\frac12\sum_{\alpha=1}^mk^\alpha_{ij}e_j\eta_\alpha\psi=0.
\end{equation*}
Next, we analyze this equation to prove Theorem \ref{thm rigidity 1} and Theorem \ref{thm rigidity 2}.

\begin{lemma}
    \label{u1}
Let $1 \leq m \leq n$, and suppose $\psi^\infty=(\psi^\infty_1,\psi^\infty_2,\dots, \psi^\infty_{2^m}) \in \overline{\mathcal S}(\mathbb{R}^n)$. 
For an integer $2\le l\le 2^m$, consider the binary expansion $2(l-1) =2^{\alpha_1}+\cdots+2^{\alpha_j}$ with $1\le\alpha_1<\cdots<\alpha_j\le m$. Then,
\begin{align}\label{equation psi infty spinor}
e^\infty_\alpha\eta_\alpha \psi^\infty=\psi^\infty
\end{align}
holds for each $1\le \alpha\le m$ if and only if 
\begin{align}\label{equation psi infty spinor 2}
\psi^\infty_l= (-1)^je^\infty_{\alpha_j}\cdots e^\infty_{\alpha_1}\psi^\infty_1\end{align}
for each $2 \leq l \leq 2^m$.

Moreover, if $\psi^{\infty}$ satisfies either~\eqref{equation psi infty spinor} or~\eqref{equation psi infty spinor 2}, then it is null for each unit timelike vector $\eta_{\alpha}$, i.e.\
$(N(\psi^\infty),X^\alpha(\psi^\infty)) \in \mathbb{R}^{n,1}$ is null.
\end{lemma}

\begin{proof}
Suppose $\psi^\infty=(v^\infty_1,v^\infty_2)$, where 
$v_i^\infty\in \mathcal{S}(\mathbb{R}^n)^{2^{m-1}}$, $i=1,2$.
Note that $e^\infty_m\eta_m\psi^\infty = \psi^\infty$ holds if and only if $v^\infty_2 = -e^\infty_m v_1^\infty$. Iterating this splitting process inductively yields the equivalence between~\eqref{equation psi infty spinor} and~\eqref{equation psi infty spinor 2}.

Now, if $\psi^{\infty}$ satisfies either~\eqref{equation psi infty spinor} or~\eqref{equation psi infty spinor 2}, then 
$e^\infty_\alpha \eta_\alpha \psi^\infty=\psi^\infty$ for $1\le \alpha\le m$. Let $\nu$ be a unit timelike vector field, and let the vector field $X^{\nu}(\psi^{\infty})$ be defined by $\langle X^{\nu}(\psi^{\infty}), Y \rangle = \langle Y \nu \psi^{\infty}, \psi^{\infty}\rangle$ so that $X^{\eta_{\alpha}}(\psi^{\infty}) = X^{\alpha}(\psi^{\infty})$ in particular. If $\nu = \sum_{\alpha=1}^m c^{\alpha} \eta_{\alpha}$, then $c^{\alpha} e^\infty_{\alpha}$ is also a unit vector and 
\[
|X^{\nu}(\psi^{\infty})| \geq
|\langle (c^{\alpha} e^\infty_{\alpha}) \nu\,\psi^\infty,\psi^\infty\rangle| = \Big| \sum_{\alpha=1}^m(c^{\alpha})^2\langle e^\infty_{\alpha} \eta_{\alpha}\psi^{\infty}, \psi^{\infty}\rangle \Big|
= |\psi^\infty|^2.
\]
Since $|X^{\nu}(\psi^{\infty})| \leq N(\psi^{\infty})$, we conclude that $N(\psi^{\infty}) = |X^{\alpha}(\psi^{\infty})|$, which means $\psi^{\infty}$ is null for each unit timelike $\eta_{\alpha}$.
\end{proof}

We now prove a theorem which is based on Theorem \ref{thm integral identity} and the proof of Theorem \ref{thm main}.

\begin{theorem}\label{l.i spinor}
    In the setting of Theorem \ref{thm main}, suppose that
    \begin{align*}
        E= \|\mathcal P\|_{\mathrm{tr}}.
    \end{align*}
    Then for any nonvanishing constant spinor $\psi^\infty_1\in \mathcal{S}(\mathbb{R}^n)$, there exists a spinor $\psi\in\overline{\mathcal S}(M)$ solving 
    \begin{equation}\label{equation spacetime harmonic spinor}
    \nabla_i\psi+\frac12\sum_{\alpha=1}^mk^\alpha_{ij}e_j\eta_\alpha\psi=0
\end{equation}
with $\psi\to \psi^\infty$ at infinity for the null spinor $\psi^\infty$ satisfying $\overline{e}^\infty_\alpha\eta_\alpha \psi^\infty=\psi^\infty$,  where $\overline{e}_\alpha^\infty=-(\mathcal U^{-1})_{\beta \alpha}\widetilde{e}^\infty_\beta$.
\end{theorem}

\begin{proof}
 Using the notation for $\mathcal{U}$, $\widetilde{e}^\infty_\alpha$ and $\widetilde{\eta}_\alpha$ from the proof of Theorem \ref{thm main}, we let $\overline{e}^\infty_\alpha=-(\mathcal U^{-1})_{\beta \alpha}\widetilde{e}^\infty_\beta$ so that $\widetilde{e}^\infty_\alpha\widetilde{\eta}_\alpha\psi^\infty=-\psi^\infty$ is equivalent to $\overline{e}^\infty_\alpha \eta_\alpha\psi^\infty=\psi^\infty$. Replacing $e^\infty_\alpha$ by $\overline{e}_\alpha^\infty$ in Lemma \ref{u1}, $\psi^{\infty}$ satisfies~\eqref{equation psi infty spinor} and is null. Thus,
 given $\psi_1^{\infty} \in \mathcal{S}(\mathbb{R}^n)$, define $\psi^{\infty} =(\psi_{1}^{\infty}, \dots, \psi_{2^m}^{\infty}) \in \overline{\mathcal S}(\mathbb{R}^n)$ by letting $\psi_l^{\infty}$ satisfy~\eqref{equation psi infty spinor 2}, for $2 \leq l \leq 2^m$. 
 We now apply Theorem \ref{thm integral identity} with this $\psi^{\infty}$. In particular, this gives a spinor $\psi \in \overline{\mathcal S}(M)$ with $\psi \to \psi^{\infty}$ at infinity. Also, since $E= \|\mathcal P\|_{\mathrm{tr}}$, we find, by the same argument as in the proof of Theorem \ref{thm main}, that 
  \begin{align*}
      EN(\psi^\infty)
+\sum_{\alpha=1}^m\langle P^\alpha,X^\alpha(\psi^\infty)\rangle
= 0. 
\end{align*}
Then, by the dominant energy condition and the integral identity of Theorem \ref{thm integral identity}, we find that $\overline{\nabla} \psi = 0$ which gives~\eqref{equation spacetime harmonic spinor}.
\end{proof} 

 For index sets $I\subset \{1,...,n\}$ and $\Gamma\subset\{1,...,m\}$, $I_j$ and $\Gamma_\alpha$ are defined as follows:
\begin{equation*}
    I_j:=
    \begin{cases}
         I \cup \{j\},      & \text{if } j \notin I, \\
         I \setminus \{j\}, & \text{if } j \in I,
    \end{cases}
    \qquad \text{and} \qquad
    \Gamma_\alpha:=
    \begin{cases}
         \Gamma \cup \{\alpha\},      & \text{if } \alpha \notin \Gamma, \\
         \Gamma \setminus \{\alpha\}, & \text{if } \alpha \in \Gamma.
    \end{cases}
\end{equation*}

Define $$\omega^\Gamma_I=\langle e_I\eta_\Gamma \psi,\psi\rangle e_I\eta_\Gamma$$ when index sets  $\Gamma$ and $I$ satisfy $|I|=|\Gamma|$, and define 
$$\omega^{\alpha\beta}_{lj}=\langle e_le_j\eta_\alpha\eta_\beta \psi,\psi\rangle$$ when $l\neq j$ and $\alpha\neq\beta$. Otherwise, we set both to zero.

\begin{lemma}\label{lemma NXw}
In the setting of Theorem \ref{l.i spinor} with $\psi$ satisfying~\eqref{equation spacetime harmonic spinor},
\begin{equation}\label{equation NXw}
   \begin{split}
       \nabla_iN=&-k^\alpha_{ij}X^\alpha_j,\\
    \nabla_iX^\alpha_j=&-Nk^{\alpha}_{ij}-k_{il}^\beta \omega_{lj}^{\alpha\beta},\\
\nabla_i\omega^\Gamma_I=& -k^\alpha_{ij}e_j\eta_\alpha \omega^{\Gamma_\alpha}_{I_j}.
   \end{split} 
\end{equation}
\end{lemma}

\begin{proof}
For the first two identities, we compute
\begin{equation*}
    \nabla_i N=-\frac{1}{2}\langle k^\alpha_{ij}e_j\eta_\alpha \psi,\psi\rangle-\frac{1}{2}\langle  \psi,k^\alpha_{ij}e_j\eta_\alpha\psi\rangle=-k^{\alpha}_{ij}X^\alpha_j
\end{equation*}
and
\begin{align*}
    \nabla_iX^\alpha_j=&-\frac{1}{2}\langle k^\beta_{il}e_j\eta_\alpha e_l\eta_\beta \psi,\psi\rangle
    -\frac{1}{2}\langle e_j\eta_\alpha\psi,k^\beta_{il}e_l\eta_\beta \psi\rangle
    \\=&-Nk^{\alpha}_{ij}-k_{il}^\beta \omega_{lj}^{\alpha\beta}.
\end{align*}
To verify the last identity, we first note that
\begin{equation} \label{omega1}
    \nabla_i \omega^\Gamma_I=-\frac{1}{2}k^{\alpha}_{ij}\left[\langle e_I\eta_\Gamma e_j\eta_\alpha\psi,\psi\rangle
    +\langle e_I\eta_\Gamma \psi,e_j\eta_\alpha\psi\rangle
    \right] e_I\eta_\Gamma.
\end{equation}
If $j\in I$ and $\alpha\notin \Gamma$, or 
$j\notin I$ and $\alpha\in \Gamma$, then $e_j\eta_\alpha$ anticommutes with $e_I \eta_\Gamma$, we obtain
\begin{equation*}
    \langle e_I\eta_\Gamma e_j\eta_\alpha\psi,\psi\rangle +\langle e_I\eta_\Gamma \psi,e_j\eta_\alpha\psi\rangle=0.
\end{equation*}
If $j\notin I$ and $\alpha\notin \Gamma$, or $j\in I$ and $\alpha\in \Gamma$, then $e_j\eta_\alpha$ commutes with $e_I \eta_\Gamma$,
we have 
\begin{align*}
    (\langle e_I\eta_\Gamma e_j\eta_\alpha\psi,\psi\rangle& +\langle e_I\eta_\Gamma \psi,e_j\eta_\alpha\psi\rangle) e_I \eta_\Gamma
    \\=&
    2\langle e_I\eta_\Gamma e_j\eta_\alpha\psi,\psi\rangle e_I \eta_\Gamma e_j\eta_\alpha e_j\eta_\alpha
    \\= &2\langle e_{I_j}\eta_{\Gamma_\alpha}\psi,\psi\rangle e_{I_j}\eta_{\Gamma_\alpha}e_j\eta_\alpha
   \\ =&2e_j\eta_\alpha\langle e_{I_j}\eta_{\Gamma_\alpha}\psi,\psi\rangle e_{I_j}\eta_{\Gamma_\alpha}.
\end{align*}
Hence, the claim follows from equation \eqref{omega1}.
\end{proof}

\begin{lemma} \label{decay1}
In the setting of Theorem \ref{l.i spinor} with $\psi$ satisfying~\eqref{equation spacetime harmonic spinor}, $N$ is bounded and nowhere vanishing, and $X^\alpha$ and $\omega^{\Gamma}_I$ are bounded. Moreover, $X^\alpha-X^\alpha (\psi^\infty)=O(r^{-q})$ and  $|\omega^{\Gamma}_I-\omega^{\Gamma}_I(\psi^\infty)|=O(r^{-q})$, where the norm $|\cdot|$ is the operator norm acting on the spinor bundle.
\end{lemma}
\begin{proof}
The first equation of \eqref{equation NXw} implies $|\nabla_i N|\le CN$ for some constant $C=C(k^1,\dots, k^m)>0$. Due to the asymptotics of $N$, we conclude that $N$ is bounded on $M$. By~\eqref{equation NXw}, $X^\alpha$ and $\omega^\Gamma_I$ are bounded.

Due to the asymptotic behavior at infinity of $N$, $N$ cannot vanish everywhere. Since $N\ne0$ somewhere, $|\nabla_i N|\le CN$ implies that $N$ is nowhere vanishing.

     Note that $\psi-\psi^\infty=O(r^{-q})$. Then, 
    \[X^\alpha-X^\alpha(\psi^\infty)=\langle e_i\eta_\alpha \psi,\psi\rangle e_i-\langle e^\infty_i\eta_\alpha \psi^\infty,\psi^\infty\rangle_{g_{\mathbb{R}^n}} e^\infty_i=O(r^{-q}).\]
    The estimates for $\omega^\Gamma_I$ follow similarly.
\end{proof}

\begin{theorem} \label{Nomega X}
For multi-indices $\Gamma$ and $I$ satisfying $|\Gamma|=|I|=\ell$, let 
\begin{equation*}
    \widetilde{X}^{\Gamma}_I:=(-1)^{\frac{\ell(\ell-1)}{2}}\sum_{\sigma\in S_\Gamma} X_{I}^{\sigma(\Gamma)}e_{I}\eta_{\sigma(\Gamma)},
\end{equation*}
where $I=\{i_1,\dots, i_\ell\}$,  $X^{\sigma(\Gamma)}_{I}=X^{\sigma_1}_{i_1}\cdots X^{\sigma_{\ell}}_{i_{\ell}}$, $S_\Gamma$  is the permutation group of $\Gamma$. 
Then, we have
    \[N^{\ell-1}\omega^\Gamma_I =\widetilde{X}^\Gamma_I.\]
    In particular, $N\omega^{\alpha\beta}_{lj}=X^\alpha_jX_l^{\beta}-X_{l}^{\alpha}X_{j}^{\beta}$.
\end{theorem}
\begin{proof}
The last equation of \eqref{equation NXw} implies
\begin{equation} \label{ODE Nomega} 
    \nabla_i(N^{\ell-1}\omega^\Gamma_I)=-(\ell-1)N^{\ell-2}k^\alpha_{ij}X^\alpha_j\omega^\Gamma_I -\sum_{\substack{j\in I, \alpha\in \Gamma\\ \text{or }j\notin I, \alpha\notin \Gamma} }N^{\ell-1}k^\alpha_{ij}e_j\eta_\alpha \omega^{\Gamma_\alpha}_{I_j}. 
\end{equation}
Let $\widehat{\sigma}_{j}(\Gamma)=\{\sigma_1,\dots,\sigma_{j-1},\sigma_{j+1},\dots, \sigma_\ell\}$.
Applying the second equation of \eqref{equation NXw},
\begin{equation} \label{tilde X ODE}
    \nabla_i \widetilde{X}^\Gamma_I= (-1)^{\frac{\ell(\ell-1)}{2}}\sum_{\sigma} \sum_{j=1}^\ell X^{\widehat{\sigma}_{j}(\Gamma)}_{I_{i_j}}\left(-k_{ii_j}^{\sigma_j} N-k_{il}^\beta\omega^{\sigma_j\beta}_{l i_j}\right)
    e_{I}\eta_{\sigma(\Gamma)}.
\end{equation}
Replacing $\omega^{\sigma_j\beta}_{l i_j}$ with $N^{-1}(X^{\sigma_j}_{i_j}X^{\beta}_l-X^{\sigma_j}_l X^{\beta}_{i_j})$ on the right-hand side of Equation \eqref{tilde X ODE}, not assuming that these two quantities are necessarily equal, we obtain
\begin{equation} \label{expression}
    (-1)^{\frac{\ell(\ell-1)}{2}}\sum_{\sigma} \sum_{j=1}^\ell X^{\widehat{\sigma}_{j}(\Gamma)}_{I_{i_j}}\left(-k_{ii_j}^{\sigma_j} N-k_{il}^\beta N^{-1}X^{\sigma_j}_{i_j}X^{\beta}_l+k_{il}^\beta N^{-1}X^{\sigma_j}_l X^{\beta}_{i_j}\right)
    e_{I}\eta_{\sigma(\Gamma)}.
\end{equation}
 Now, we analyze each term in Equation \eqref{expression} and later compare the results with those in Equation \eqref{ODE Nomega}. First, we have
\begin{equation*}
    (-1)^{\frac{\ell(\ell-1)}{2}}\sum_{\sigma} \sum_{j=1}^\ell -X^{\widehat{\sigma}_{j}(\Gamma)}_{I_{i_j}}k_{ii_j}^{\sigma_j} N e_I \eta_{\sigma(\Gamma)}
    = -\sum_{j\in I,\alpha\in \Gamma}
    N k^\alpha_{ij} e_j\eta_\alpha \widetilde{X}_{I_j}^{\Gamma_\alpha}.
\end{equation*}
For the second term, we obtain
\begin{equation*}
    (-1)^{\frac{\ell(\ell-1)}{2}}\sum_{\sigma} \sum_{j=1}^\ell -X^{\widehat{\sigma}_{j}(\Gamma)}_{I_{i_j}} k^\beta_{il} N^{-1}X^{\sigma_j}_{i_j}X^\beta_l e_I\eta_{\sigma(\Gamma)}
    = -\ell N^{-1} k^\beta_{il}X^\beta_l \widetilde{X}_{I}^\Gamma.
\end{equation*}
To handle the last term in Equation \eqref{expression}, we consider several cases.
For $\beta\in \Gamma$ with $\beta=\alpha_s$ and $\beta\neq \sigma_j$, let $\varsigma\in S_\Gamma$ be the permutation that swaps $\alpha_j$ and $\alpha_s$. Since $\eta_{\sigma\circ \varsigma(\Gamma)}=-\eta_{\sigma(\Gamma)}$, we have
\begin{align*}
    X^{\widehat{\sigma}_{j}(\Gamma)}_{I_{i_j}}  X^{\alpha_s}_{i_j}X^{\sigma_j}_l e_I\eta_{\sigma(\Gamma)}+X^{\widehat{(\sigma\circ \varsigma)}_{s}(\Gamma)}_{I_{i_s}}  X^{\alpha_s}_{i_s}X^{(\sigma\circ\varsigma)_s}_l e_I\eta_{\sigma\circ \varsigma(\Gamma)}=0.
\end{align*}
Then, we obtain
\begin{align*}
(-1)^{\frac{\ell(\ell-1)}{2}}\sum_{\sigma} \sum_{j=1}^\ell \sum_{\beta\in\Gamma,\beta\neq \sigma_j}  X^{\widehat{\sigma}_{j}(\Gamma)}_{I_{i_j}} k_{il}^\beta N^{-1}  X^\beta_{i_j}X^{\sigma_j}_l e_I\eta_{\sigma(\Gamma)}=0.
\end{align*}
When $\beta=\sigma_j$, we have 
\begin{equation*}(-1)^{\frac{\ell(\ell-1)}{2}}\sum_{\sigma} \sum_{j=1}^\ell X^{\widehat{\sigma}_{j}(\Gamma)}_{I_{i_j}} k^{\sigma_j}_{il} X^{\sigma_j}_{i_j}X^{\sigma_j}_l e_I\eta_{\sigma(\Gamma)}
    = \sum_{\beta\in \Gamma} k^\beta_{il}X^\beta_l \widetilde{X}_{I}^\Gamma.
\end{equation*}
Suppose $\beta\notin \Gamma$ and $l:=i_s\in I$ and $l\neq i_j$.
In this case,
\begin{align*}
      X^{\widehat{\sigma}_{j}(\Gamma)}_{I_{i_j}}  X^{\beta}_{i_j}X^{\sigma_j}_{i_s} e_I\eta_{\sigma(\Gamma)}+ X^{\widehat{(\sigma\circ\varsigma)}_{j}(\Gamma)}_{I_{i_j}}  X^{\beta}_{i_j}X^{(\sigma\circ \varsigma)_j}_{i_s} e_I\eta_{\sigma\circ \varsigma(\Gamma)}=0.
\end{align*}
It implies that
\begin{equation*}
   (-1)^{\frac{\ell(\ell-1)}{2}}\sum_{\sigma} \sum_{j=1}^\ell \sum_{\beta\notin \Gamma}\sum_{l\neq i_j, l\in I}  X^{\widehat{\sigma}_{j}(\Gamma)}_{I_{i_j}} k_{il}^\beta N^{-1}  X^\beta_{i_j}X^{\sigma_j}_l e_I\eta_{\sigma(\Gamma)}
   =0.
\end{equation*}
If $\beta\notin \Gamma$ and $l=i_j$, we have 
\begin{equation}
(-1)^{\frac{\ell(\ell-1)}{2}}\sum_{\sigma} \sum_{j=1}^\ell\sum_{\beta\notin \Gamma} X^{\widehat{\sigma}_{j}(\Gamma)}_{I_{i_j}} k^{\beta}_{ii_j} N^{-1} X^{\beta}_{i_j}X^{\sigma_j}_{i_j} e_I\eta_{\sigma(\Gamma)}
    = \sum_{\beta\notin \Gamma,l\in I}k^\beta_{il} N^{-1}X^\beta_l \widetilde{X}_{I}^\Gamma.
\end{equation}

When $\beta\notin \Gamma$ and $l\notin\Gamma$, we obtain 
\begin{equation*} (-1)^{\frac{\ell(\ell-1)}{2}}\sum_{\sigma} \sum_{j=1}^\ell \sum_{\beta\notin \Gamma,l\notin I}  X^{\widehat{\sigma}_{j}(\Gamma)}_{I_{i_j}} k_{il}^\beta N^{-1} X^\beta_{i_j}X^{\sigma_j}_l e_I\eta_{\sigma(\Gamma)}
 =\sum_{\beta\notin \Gamma, l\notin I}N^{-1}(-k^\beta_{il} e_l \eta_\beta \widetilde{X}^{\Gamma_\beta}_{I_l}+k^{\beta}_{il}X^\beta_l\widetilde{X}^\Gamma_I).
\end{equation*}
Therefore, combining all the equations above, we have
\begin{equation} \label{ODE F0}
\begin{split}
     &\nabla_i (N^{\ell-1}\omega_I^\Gamma-\widetilde{X}^\Gamma_I)
    \\
    & \qquad = -(\ell-1)N^{-1}k^{\alpha}_{ij}X^\alpha_j (N^{\ell-1}\omega_I^\Gamma-\widetilde{X}^\Gamma_I)
    -\sum_{j\in I,\alpha\in \Gamma}k^\alpha_{ij}e_j\eta_\alpha N(N^{\ell-2}\omega^{\Gamma_\alpha}_{I_j}-\widetilde{X}^{\Gamma_\alpha}_{I_j})
   \\&\qquad \,\quad -\sum_{j\notin I,\alpha\notin \Gamma}k^\alpha_{ij}e_j\eta_\alpha N^{-1}(N^{\ell}\omega^{\Gamma_\alpha}_{I_j}-\widetilde{X}^{\Gamma_\alpha}_{I_j})
   \\&\qquad \,\quad +\sum_{\sigma\in S_\Gamma}\sum_{j=1}^\ell (-1)^{\frac{\ell(\ell-1)}{2}}X^{\widehat{\sigma}_j(\Gamma)}_{I_{i_j}} k^{\beta}_{il} N^{-1}(N\omega_{li_j}^{\sigma_j\beta}-X^{\sigma_j}_{i_j}X^\beta_l+X^{\sigma_j}_lX^\beta_{i_j})e_I\eta_{\sigma(\Gamma)}.
\end{split}
\end{equation}
Let 
\[F:=\sum_{\ell=1}^m\sum_{|I|=|\Gamma|=\ell}|N^{\ell-1}\omega_I^\Gamma-\widetilde{X}^\Gamma_I|^2.\]
From Equation \eqref{ODE F0},  we obtain $|\nabla_i F|\le CF$ for some $C>0$.

By Theorem \ref{l.i spinor}, we have that $N(\psi^\infty)=|X^\alpha(\psi^\infty)|_{g_{\mathbb{R}^n}}$. We extend $\{\overline{e}^\infty_\alpha\}_{\alpha=1}^m$ to an orthonormal basis $\{\overline{e}^\infty_i\}_{i=1}^n$ of $\mathbb{R}^n$.
If there exists $\alpha\in \Gamma\setminus I$, then $\overline{e}^\infty_\alpha\eta_\alpha$ anti-commutes with $\overline{e}^\infty_I\eta_\Gamma$. Thus, $\omega^\Gamma_I(\psi^\infty)=0$ and $\widetilde{X}(\psi^\infty)=0$. If $I=\Gamma$, then $\overline{e}^\infty_\alpha\eta_\alpha\psi^\infty=\psi^\infty$ implies that   $N^{|I|-1}(\psi^\infty)\omega^\Gamma_I(\psi^\infty)=\widetilde{X}(\psi^\infty)$.
Therefore, $|N^{\ell-1}\omega_I^\Gamma-\widetilde{X}_I^\Gamma|=O(r^{-q})$ by Lemma \ref{decay1}, and we have $F\equiv 0$. 
This completes the proof.
\end{proof}

\begin{theorem} \label{N=X}
In the setting of Theorem \ref{l.i spinor} with $\psi$ satisfying~\eqref{equation spacetime harmonic spinor}, $N=|X^\alpha|$ for any $\alpha=1,...,m$, i.e.\ $\psi$ is everywhere null. Moreover, $X^\alpha$ is perpendicular to $X^\beta$, for any $\alpha\neq \beta$.
\end{theorem}

\begin{proof}
Using Lemma \ref{lemma NXw} and Theorem \ref{Nomega X}, we have
\begin{equation}
\begin{split}
    &\frac{1}{2}\nabla_i \sum_\alpha(N^2-|X^\alpha|^2)
    \\&\quad = -mNk^{\alpha}_{ij}X^\alpha_j+Nk^{\alpha}_{ij}X^\alpha_{j}+k^{\beta}_{il}\omega^{\alpha\beta}_{lj}X^\alpha_j
    \\&\quad =-(m-1)Nk^{\alpha}_{ij}X^\alpha_j+k^{\beta}_{il}N^{-1}(X^\alpha_j X^\beta_l-X^\alpha_l X^\beta_j)X^\alpha_j
    \\&\quad = \sum_{\alpha\neq \beta}k^{\beta}_{ij}X^\beta_j N^{-1}(|X^\alpha|^2-N^2)-\sum_{\alpha\neq \beta}k^\beta_{il}N^{-1}X^\alpha_l \langle X^\alpha, X^\beta\rangle.
\end{split}
\end{equation}
Moreover, 
\begin{equation}
    \begin{split}
        \nabla_i & \sum_{\alpha\neq \beta}\langle X^\alpha, X^\beta\rangle
      \\&\quad =\sum_{\alpha\neq \beta}
      (-Nk^\alpha_{ij}-k_{il}^\gamma\omega_{lj}^{\alpha\gamma})X_j^\beta+(-Nk^\beta_{ij}-k_{il}^\gamma\omega_{lj}^{\beta\gamma})X_j^\alpha
      \\& \quad= \sum_{\alpha, \,\beta, \gamma \text{ distinct}} k^{\gamma}_{il}N^{-1}\left(2\langle X^\beta,X^\gamma\rangle X^\alpha_l-2\langle X^\alpha, X^\beta\rangle X^\gamma_l\right)
      \\&\quad \quad +\sum_{\alpha\neq \beta}\left[2k^{\alpha}_{il}X^\beta_l N^{-1}(|X^\alpha|^2-N^2)-2\langle X^\alpha,X^\beta\rangle k^\beta_{il}X^\beta_l\right]
    \end{split}
\end{equation}
Let 
\[Q=\sum_\alpha (N^2-|X^\alpha|^2)^2+\sum_{\alpha\neq\beta}\langle X^\alpha, X^\beta\rangle^2.\]
Using the above formulas, together with $|X^\alpha| \le N$, we obtain the ODE $|\nabla Q|\le CQ$. 
Combining this with the decay estimate $Q=O(r^{-q})$, we have $Q\equiv 0$ which yields $N=|X^\alpha|$.
\end{proof}
From now on, by Theorem \ref{N=X}, we choose $e_\alpha = N^{-1}X^\alpha$ and extend $\{e_\alpha\}_{\alpha=1}^m$ to an orthonormal frame $\{e_i\}_{i=1}^n$.
\begin{proposition} \label{Sigma}
    In the setting of Theorem \ref{l.i spinor} with $\psi$ satisfying~\eqref{equation spacetime harmonic spinor}, the vector fields $\{X^\alpha\}_{\alpha=1}^m$ are normal to a codimension $m$ foliation $\{\Sigma^{n-m}\}$, whose second fundamental form vector is $-N^{-1}k^\alpha_{\textsc{AB}}X^\alpha$. 
\end{proposition}

\begin{proof} 
By Theorem \ref{N=X},  $\{e_{\textsc{A}}\}_{\textsc{A}=m+1}^n$ is an orthonormal basis of the orthogonal complement of the subspace $\operatorname{Span}\{X^1,\dots,X^m\}$. 
    Combining the identity $N\omega_{lj}^{\alpha\beta}=X^\alpha_j X_l^\beta-X^\beta_j X_l^\alpha$ from Theorem \ref{Nomega X} with the second equation of \eqref{equation NXw} yields
    \begin{equation} \label{aXb}
        \nabla_\textsc{A}X^\alpha_\textsc{B} = -Nk^{\alpha}_{\textsc{AB}}, \qquad \text{where } \textsc{A},\textsc{B}=m+1,\dots,n.  
    \end{equation}
    Hence, $\nabla_\textsc{A}X^\alpha_\textsc{B} - \nabla_\textsc{B}X^\alpha_\textsc{A} = 0$. By \cite[Proposition 19.8]{JohnLee}, the distribution $\operatorname{Span}\{e_{m+1},\dots,e_n\}$ is involutive. Thus, it induces a codimension $m$ foliation with unit normal vectors $e_\alpha$.  
    Moreover, equation \eqref{aXb} implies that the second fundamental form vector is $-N^{-1}k^\alpha_{\textsc{AB}}X^\alpha$.
\end{proof}

Let $\phi \in \mathcal{S}(M^n)$. Note that $\mathcal{S}(M^n)|_{\Sigma}$ is a principal $\text{Spin}(n-m)$-bundle. By abuse of notation, we write $\phi$ also for its restriction $\phi|_\Sigma$.

\begin{lemma} \label{A phi}
    Let $h^{\alpha}_{AB}e_\alpha$ be the second fundamental form vector of $\Sigma$. Then,
    \[
        \nabla^\Sigma_A \phi = \nabla_A\phi + \frac{1}{2}h^\alpha_{AB}e_B e_\alpha\phi - \frac{1}{4}\langle \nabla_A e_\alpha, e_\beta\rangle e_\alpha e_\beta \phi.
    \]
\end{lemma}

\begin{proof}
    Since
    \[
        \nabla_{i} \phi = e_i(\phi) + \frac{1}{4} \langle \nabla_{e_i} e_j, e_l\rangle e_j e_l \phi,
    \]
    we obtain
    \begin{align*}
        \nabla^{\Sigma}_A \phi &= e_A(\phi) + \frac{1}{4}\langle \nabla^{\Sigma}_A e_B, e_C\rangle e_B e_C\phi \\
        &= \nabla_A\phi - \frac{1}{2}\langle \nabla_A e_\alpha, e_B\rangle e_\alpha e_B \phi
            - \frac{1}{4}\langle \nabla_A e_\alpha, e_\beta\rangle e_\alpha e_\beta \phi.
    \end{align*}
    Note that $\langle \nabla_A e_\alpha, e_B\rangle = h^\alpha_{AB}$, so the claimed identity follows.
\end{proof}
 Let $\psi=(\psi_1,...,\psi_{2^m})$, where $\psi_i\in \mathcal{S}(M^n)$. 
\begin{lemma} \label{parallel}
 The spinor $N^{-\frac{1}{2}}\psi_1$ is parallel on $\Sigma$.
\end{lemma}

\begin{proof}
    By Proposition~\ref{Sigma} and $e_\alpha=N^{-1}X^\alpha$, we have $h^\alpha_{AB} = -k^{\alpha}_{AB}$.   Equation \eqref{equation NXw} together with $N^{-1}\omega^{\alpha\beta}_{lj} = g_{\alpha j}g_{\beta l} - g_{\alpha l}g_{\beta j}$ by Theorem \ref{Nomega X} imply
    \begin{align*}
        \langle \nabla_A e_\alpha, e_\gamma \rangle &= \langle N^{-1} \nabla_A X^\alpha, e_\gamma \rangle - N^{-1}(\nabla_A N)g_{\alpha\gamma} \\
        &= -k^\alpha_{A\gamma} - k^{\beta}_{A\beta}g_{\alpha\gamma} + k^\gamma_{A\alpha} + k^\beta_{A\beta}g_{\alpha\gamma} \\
        &= -k^\alpha_{A\gamma} + k^\gamma_{A\alpha}.
    \end{align*}
    Applying Lemma~\ref{A phi} and $(\eta_\alpha \psi)_1=-e_\alpha\psi_1$ yields
    \begin{align*}
        \nabla^\Sigma_A\psi_1 &= \nabla_A \psi_1 + \frac{1}{2}h^\alpha_{AB}e_B e_\alpha\psi_1 - \frac{1}{4}\langle \nabla_A e_\alpha, e_\gamma\rangle e_\alpha e_\gamma \psi_1 \\
        &= -\frac{1}{2}k^\alpha_{Aj}e_j(\eta_\alpha\psi)_1 - \frac{1}{2}k^{\alpha}_{AB}e_Be_\alpha\psi_1
            - \frac{1}{4}\bigl(k^\gamma_{A\alpha} - k^{\alpha}_{A\gamma}\bigr)e_\alpha e_\gamma\psi_1 \\
        &= \frac{1}{2}k^\alpha_{Aj}e_j e_\alpha\psi_1 - \frac{1}{2}k^{\alpha}_{AB}e_Be_\alpha\psi_1
            - \sum_{\gamma \neq \alpha} \frac{1}{2}k^{\alpha}_{A\gamma}e_\gamma e_\alpha \psi_1 \\
        &= -\frac{1}{2} k^\alpha_{A\alpha}\psi_1.
    \end{align*}
    Consequently,
    \[
        \nabla^\Sigma_A \bigl(N^{-\frac{1}{2}}\psi_1\bigr) = - \frac{1}{2}k^\alpha_{A\alpha}N^{-\frac{1}{2}}\psi_1 + \nabla^\Sigma_A \bigl(N^{-\frac{1}{2}}\bigr)\psi_1 = 0. \qedhere
    \]
\end{proof}

\begin{lemma} \label{NX u psi}
    For any  spinor $u\in \overline{\mathcal{S}}(M^n)$ 
    asymptotic to a constant spinor $u^\infty\in \overline{\mathcal{S}}(\mathbb{R}^n)$ that satisfies $\overline{e}^\infty_\alpha\eta_\alpha u^\infty=u^\infty$ for all $1\le \alpha\le m$, $\|u^\infty\|_{g_{\mathbb{R}^n}}=\|\psi^\infty\|_{g_{\mathbb{R}^n}}$
    and
    $\nabla_i u = -\frac{1}{2}k^\alpha_{ij}e_j\eta_\alpha u$, 
    we have that $N(u) = N(\psi)$ and $X^\alpha(u) = X^\alpha(\psi)$ hold everywhere for every $\alpha = 1,\dots,m$.
\end{lemma}
\begin{proof}
Note that $N(u)$ and $\omega^\Gamma_I(u)$ obey the ODE system \eqref{equation NXw}. Combining this with the decay estimates in Lemma~\ref{decay1}, it suffices to show $\omega^\Gamma_I(\psi^\infty) = \omega^\Gamma_I(u^\infty)$ in order to prove $\omega^\Gamma_I(u) = \omega^\Gamma_I(\psi)$.

 By construction, we have $N(u^\infty)=N(\psi^\infty)$ and $X^\alpha(u^\infty)=N(u^\infty)\overline{e}^\infty_\alpha=X^\alpha(\psi^\infty)$.
Following the argument at the end of the proof in Theorem \ref{Nomega X}, we have $N^{\ell-1}(u^\infty)\omega^\Gamma_I(u^\infty)=\widetilde{X}^\Gamma_I(u^\infty)$. Consequently, $\omega^\Gamma_I(\psi^\infty) = \omega^\Gamma_I(u^\infty)$, which completes the proof.     
\end{proof}

\begin{proposition}\label{proposition flat}
Each leaf $\Sigma$ of the foliation is flat.
\end{proposition}

\begin{proof}
By Theorem \ref{l.i spinor}, for any constant nonvanishing spinor
$u_1^\infty \in \mathcal{S}(\mathbb{R}^n)$, there exists
$u \in \overline{\mathcal{S}}(M^n)$ such that
\[
\nabla_i u = -\frac{1}{2} k^\alpha_{ij} e_j \eta_\alpha u.
\]
Combined with Lemma~\ref{NX u psi}, all such spinors $u$ correspond to the
same $N$ and $X^\alpha$. Consequently, $e_\alpha \eta_\alpha u = u$,
and they induce the same foliation.

Moreover, by Lemma \ref{parallel}, the spinor $N^{-1/2} u_1$ is parallel on
$\Sigma$. By choosing different asymptotic spinors $u_1^\infty$, we obtain
$2^n$ spinors in $\mathcal{S}(M)$ that are parallel on $\Sigma$ and linearly
independent at each point of $M^n$. 
Therefore, $R^{\Sigma}_{ABCD}=0$, i.e., $\Sigma$ is flat.
\end{proof}
\begin{proof}[Proof of Theorem \ref{thm rigidity 1}]
    The statement of Theorem \ref{thm rigidity 1} follows from Proposition \ref{Sigma} combined with Proposition \ref{proposition flat}.
\end{proof}

\subsection{Multiple Killing developments}

\begin{proof}[Proof of Theorem \ref{thm rigidity 2}]
On $\overline M= M\times \mathbb R^m$, define
\begin{align}\label{eq: Killing development}
      \overline{g}=-N^2\sum_{\alpha=1}^md\tau^2_\alpha+g_{ij}(dx_i+X_i^\alpha d\tau_\alpha)(dx_j+X_j^\beta d\tau_\beta).
\end{align}
Note that the timelike unit normal vectors are given by
\begin{align*}
    \eta_\alpha=N^{-1}(\partial_{\tau_\alpha}-X^\alpha).
\end{align*}
We have
\begin{align*}
    \langle \overline\nabla_i\eta_\alpha, e_j\rangle=&
    N^{-1}\langle \overline \nabla_i(\partial_{\tau_\alpha}-X^\alpha),e_j\rangle\\
    =&N^{-1}\overline\Gamma_{i\alpha}^k\overline g_{kj}-N^{-1}\nabla_iX_j^\alpha-N^{-1}X_k^\alpha \overline \Gamma_{ij}^k\\
       =&\frac12 N^{-1}(\nabla_i \overline g_{\alpha j}+\nabla_\alpha \overline g_{ij}-\nabla_j\overline g_{i\alpha}))-N^{-1}\nabla_iX_j^\alpha\\
    =&-\frac12N^{-1}(\nabla_iX^\alpha_j+\nabla_jX^\alpha_i)\\
    =&k_{ij}^\alpha.
\end{align*}
Note that the last equation crucially uses that $k^\beta = f^\beta g$ for all $\beta\ne\alpha$.
In the Killing development, we extend $\psi$ trivially and denote this new spinor by $\overline \psi$, i.e., $\partial_{\tau_\alpha}\overline{\psi}=0$ for $\alpha=1,...,m$.
We obtain
\begin{align*}
    \overline\nabla_i\overline\psi=0,\qquad \overline\nabla_\alpha\overline\psi=0.
\end{align*}
The verification of these identities is standard although tedious and will be carried out in Appendices \ref{appendix connection formula} and \ref{appendix connection coefficients}. These imply that $\overline{\psi}$ is a nontrivial parallel null spinor, i.e. $(\overline{M}^{n+m}, \overline{g})$ is a generalized $pp$-wave.
\end{proof}

\section{Spinorial symmetries}\label{section spinorial symmetries}

Let $\psi$ be a spinor solving 
\begin{align*}
     \nabla_i\psi=-\frac12\sum_{\alpha=1}^m k^\alpha_{ij}e_j\eta_\alpha\psi.
\end{align*}
In this section, we show that we automatically obtain a second spinor solving the same PDE:

\begin{theorem} \label{ODE varphi}
Let $\mathcal{M}=\{1,2,...,m\}$. 
Define 
\[\varphi=\eta_\mathcal{M} \sum_{|I|=|\Gamma|} (-1)^{|\Gamma|}\langle e_I \eta_\Gamma\psi,\psi\rangle e_I \eta_\Gamma\psi.\]
Then, $\varphi $ also solves
\begin{equation*}
    \nabla_i\varphi+\frac12\sum_{\alpha=1}^mk^\alpha_{ij}e_j\eta_\alpha\varphi=0.
\end{equation*}
\end{theorem}

\begin{proof}
Applying equation \eqref{equation NXw}, we have
\begin{equation}\label{ivarphi}
    \begin{split}
         \nabla_i\varphi &=
    \eta_{\mathcal{M}}\nabla_i\left(\sum_{|I|=|\Gamma|}(-1)^{|\Gamma|}\omega_I^\Gamma\psi\right)
    \\& = -\eta_{\mathcal{M}}\sum_{|I|=|\Gamma|}\sum_{j,\alpha} (-1)^{|\Gamma|}k_{ij}^\alpha\Big(e_j\eta_\alpha\omega^{\Gamma_\alpha}_{I_j}+\frac{1}{2}\omega^{\Gamma}_I e_j\eta_\alpha\Big)\psi.
    \end{split}
\end{equation}
If $j\in I$ and $\alpha\notin \Gamma$, or 
$j\notin I$ and $\alpha\in \Gamma$, we obtain 
\begin{equation*}
    \eta_\mathcal{M} e_I\eta_\Gamma e_j\eta_\alpha=e_j\eta_\alpha\eta_\mathcal{M} e_I\eta_\Gamma.
\end{equation*}
If $j\notin I$ and $\alpha\notin \Gamma$, or $j\in I$ and $\alpha\in \Gamma$, we obtain
\begin{equation*}
    \eta_\mathcal{M} e_I\eta_\Gamma e_j\eta_\alpha=-e_j\eta_\alpha\eta_\mathcal{M} e_I\eta_\Gamma.
\end{equation*}
Therefore, the result follows after simplifying \eqref{ivarphi}.
\end{proof}

\begin{proposition}
    The spinor $\varphi $ vanishes if $\psi$ is null.
\end{proposition}

\begin{proof}
By Theorem \ref{l.i spinor},
 we have $\overline{e}^\infty_\alpha\eta_\alpha\psi^\infty=\psi^\infty$. Then 
for any $\alpha\in \Gamma\setminus I$, $\overline{e}^\infty_\alpha \eta_\alpha$ anti-commutes with $\overline{e}^\infty_I\eta_\Gamma$, thus, $\langle \overline{e}^\infty_I \eta_\Gamma\psi^\infty,\psi^\infty\rangle=0$. Therefore, 
\begin{equation*}
    \varphi^\infty=\eta_\mathcal{M}\sum_{\Gamma}(-1)^{|\Gamma|}\langle \overline{e}^\infty_\Gamma\eta_\Gamma \psi^\infty,\psi^\infty \rangle
    \overline{e}^\infty_\Gamma\eta_\Gamma\psi^\infty=\eta_\mathcal{M}(1-1)^m|\psi^\infty|^2\psi^\infty=0.
\end{equation*}
Hence, Lemma \ref{decay1} yields $\varphi=O(r^{-q})$. However, Theorem \ref{ODE varphi} implies that $|\nabla_i \varphi|\le C|\varphi|$ for some constant $C>0$, therefore, $\varphi\equiv 0$.
\end{proof}

\appendix

\section{Connection formula for spinors}\label{appendix connection formula}

The connection formula for spinors in a spacetime with multiple time directions is not standard in the literature, so we present it here. For details, see \cite[Sections 1.6 and 2.4]{lawson2016spin}. Let $(\overline{M}^{n,m},\overline{g})$ be a semi-Riemannian manifold with the Levi-Civita connection $\overline{\nabla}$. Let $\{E_\mathbf{A}\}$ be an orthonormal frame for $\overline{g}$. We denote by $E_i$ the spacelike unit vectors and by $E_\alpha$ the timelike ones.

Recall that 
\[\operatorname{Spin}(n,m)=\{v_1\cdots v_{2p}: p\in \mathbb{N},\; v_i\in T\overline{M}^{n,m} \text{  and  } \|v_i\|_{\overline{g}}=\pm 1\}.\]
Denote by $\rho:\operatorname{Spin}(n,m) \to \operatorname{SO}(n,m)$ the group homomorphism generated by 
\[\rho(E_\mathbf{A})v=-E_\mathbf{A}vE_\mathbf{A}^{-1}=v-2\frac{\langle v,E_\mathbf{A}\rangle_{\overline{g}}}{\langle E_\mathbf{A}, E_\mathbf{A}\rangle_{\overline{g}}}E_\mathbf{A}, \quad v\in T\overline{M}^{n,m}.\]
Let
$\operatorname{Lie}(\operatorname{Spin}_{n,m})$ be the Lie algebra of $\operatorname{Spin}(n,m)$. Note that
\begin{equation} \label{cos}
\begin{split}
    \cos 2t +\sin 2t E_i E_j=& (\cos t E_i+\sin t E_j)(-\cos t E_i+\sin t E_j),
    \\  \cosh 2t +\sinh 2t E_\alpha E_j=& (\cosh t  E_\alpha-\sinh t E_j)(\cosh t E_\alpha+\sinh t E_j), 
    \\ \cos 2t +\sin 2t E_\alpha E_\beta=& (\cos t E_\alpha-\sin t E_\beta)(\cos t E_\alpha+\sin t E_\beta).
\end{split}
\end{equation}
Thus, 
$\Gamma(t)$ defined by the expressions given above is a curve on $\operatorname{Spin}(n,m)$. Since $\Gamma'(0)=2E_\mathbf{A}E_\mathbf{B}$,
$\operatorname{Lie}(\operatorname{Spin}_{n,m})$ is spanned by $\{E_\mathbf{A}E_\mathbf{B}\}$, while $\operatorname{so}(n,m)$ is spanned by $E_\mathbf{A}\wedge E_\mathbf{B}$, where
\[E_\mathbf{A}\wedge E_\mathbf{B}(v):=-\langle v,E_\mathbf{B}\rangle_{\overline{g}} E_\mathbf{A}+\langle v,E_\mathbf{A}\rangle_{\overline{g}} E_\mathbf{B}.\]
Let $\rho_*: \operatorname{Lie}(\operatorname{Spin}_{n,m})\to \operatorname{so}(n,m)$ be the map on the tangent space induced by $\rho$.
According to \cite[pg.\ 42]{lawson2016spin},
\begin{align*}
    \rho_*(E_\mathbf{A}E_\mathbf{B})v=\frac{1}{2}\frac{\partial}{\partial t}\Big{|}_{t=0}\rho(\Gamma(t)) v=\frac{1}{2}\frac{\partial}{\partial t}\Big{|}_{t=0}\left[ \Gamma(t) v\Gamma(t)^{-1}\right]=E_\mathbf{A}E_\mathbf{B}v-vE_\mathbf{A}E_\mathbf{B}
    =2E_\mathbf{A}\wedge E_\mathbf{B}(v).
\end{align*}
Let $w$ be the connection 1-form of $\overline{\nabla}$, i.e., $w=\frac{1}{2}w_{\mathbf{B}\mathbf{A}}E_\mathbf{A}\wedge E_\mathbf{B}$ and $wE_\mathbf{A}=\overline{\nabla} E_\mathbf{A}$. Thus, 
\[w_{\mathbf{B}\mathbf{A}}=\frac{\langle \overline{\nabla} E_\mathbb{A}, E_\mathbf{B}\rangle_{\overline{g}}}{\langle E_\mathbf{A}, E_\mathbf{A}\rangle_{\overline{g}}\langle E_\mathbf{B}, E_\mathbf{B}\rangle_{\overline{g}}}.\]
Therefore, \[\rho_*^{-1}(w)=\frac{1}{4}\langle \overline{\nabla} E_i,E_j\rangle_{\overline{g}}  E_i E_j-\frac{1}{2}\langle \overline{\nabla} E_\alpha, E_j\rangle_{\overline{g}}  E_\alpha E_j+\frac{1}{4}\langle \overline{\nabla} E_\alpha, E_\beta\rangle_{\overline{g}}  E_\alpha E_\beta.\]
\begin{lemma}\label{connection formula}
Let $\overline{\psi}$ be a spinor on $(\overline{M}^{n,m},\overline{g})$. Then,
    \[\overline{\nabla}_v \overline{\psi}=v(\overline{\psi})+\frac{1}{4}\langle \overline{\nabla}_v E_i,E_j\rangle_{\overline{g}} E_i E_j\overline{\psi}-\frac{1}{2}\langle \overline{\nabla}_v E_\alpha, E_j\rangle_{\overline{g}} E_\alpha E_j\overline{\psi}+\frac{1}{4}\langle \overline{\nabla}_v E_\alpha, E_\beta\rangle_{\overline{g}} E_\alpha E_\beta\overline{\psi}.\]
\end{lemma}

\section{Verification of parallel spinors}\label{appendix connection coefficients}
In this appendix we verify that the spinor $\overline{\psi}$ defined in the proof of Theorem \ref{thm rigidity 2} is parallel. 
Recall that $k^\alpha=f^\alpha g$.
Thus, equation \eqref{equation NXw} implies
\begin{equation*}
    \nabla_i N=-f^\alpha X^\alpha_i,\quad 
    \nabla_i X^\alpha_j=Nf^\alpha g_{ij}-N^{-1}f^\beta(X^\alpha_jX^\beta_i-X^\alpha_i X^\beta_j).
\end{equation*}
 We will apply the above equations to compute the connection coefficients of $\overline{g}$ defined in \eqref{eq: Killing development}. Since $\eta_\alpha=N^{-1}(\partial_{\tau_\alpha}-X^\alpha)$, we obtain for the Lie brackets
\begin{align*}
    [\eta_\alpha,\eta_\beta]=&(\nabla_{X^\beta}N^{-1})\eta_\alpha-(\nabla_{X^\alpha}N^{-1})\eta_\beta+N^{-2}[X^\alpha,X^\beta]
    \\=& f^\beta \eta_\alpha-f^\alpha \eta_\beta,
    \\ [\eta_\alpha, e_i]=&-f^\beta X^\beta_i N^{-1} \eta_\alpha-N^{-1}[X^\alpha,e_i].
\end{align*}

Next, we compute the connection coefficients 
    \begin{align*}
        \langle \overline{\nabla}_{i} \eta_\alpha,\eta_\beta\rangle
        =&\frac{1}{2}\left(\langle[\eta_\beta,\eta_\alpha],e_i\rangle
        +\langle[\eta_\beta,e_i],\eta_\alpha\rangle-
        \langle[\eta_\alpha,e_i],\eta_\beta\rangle\right)=0,
    \end{align*}
     \begin{align*}
        \langle \overline{\nabla}_{\eta_\alpha} \eta_\beta,e_i\rangle
        =&\frac{1}{2}\left(\langle[e_i,\eta_\beta],\eta_\alpha\rangle
        +\langle[e_i,\eta_\alpha],\eta_\beta\rangle-
        \langle[\eta_\beta,\eta_\alpha],e_i\rangle\right)
        \\=& -f^\gamma X_i^\gamma N^{-1}\delta_{\alpha\beta},
    \end{align*}
    \begin{align*}
        \langle \overline{\nabla}_{\eta_\alpha} e_i,e_j\rangle
        =&\frac{1}{2}\left(\langle[e_j,e_i],\eta_\alpha\rangle
        +\langle[e_j,\eta_\alpha],e_i\rangle-
        \langle[e_i,\eta_\alpha],e_j\rangle\right)
        \\=& \frac{1}{2}N^{-1}(\langle \nabla_i X^\alpha, e_j\rangle-\langle \nabla_j X^\alpha, e_i\rangle)
        \\=&N^{-2}f^\beta (X^\alpha_i X^\beta_j-X^\alpha_j X^\beta_i),
    \end{align*}
    \begin{align*}
        \langle \overline{\nabla}_{\eta_\alpha} \eta_\beta,\eta_\gamma\rangle
        =&\frac{1}{2}\left(\langle[\eta_\gamma,\eta_\beta],\eta_\alpha\rangle
        +\langle[\eta_\gamma,\eta_\alpha],\eta_\beta\rangle-
        \langle[\eta_\beta,\eta_\alpha],\eta_\gamma\rangle\right)
        \\=& f^\gamma \delta_{\beta\alpha}-f^\beta \delta_{\alpha\gamma}.
    \end{align*}
    Finally,  we show that $\overline{\psi}$ is parallel. Applying Lemma \ref{A phi}, we obtain
    \begin{equation*}
        \overline{\nabla}_i \overline{\psi}=\nabla_i \overline{\psi}+\frac{1}{2}k^\alpha_{ij}e_j\eta_\alpha \overline{\psi}+\frac{1}{4}\langle \overline{\nabla}_i \eta_\alpha, \eta_\beta \rangle \eta_\alpha \eta_\beta \overline{\psi}
        = 0. 
    \end{equation*}
Using $\eta_\alpha \overline{\psi}=-N^{-1}X^\alpha \overline{\psi}$ and Lemma \ref{connection formula} yields
\begin{equation} \label{eta1}
    \begin{split}
    \overline{\nabla}_{\eta_\alpha}\overline{\psi}=&
    \eta_\alpha(\overline{\psi})
    +\frac{1}{4}\langle \overline{\nabla}_{\eta_\alpha}\eta_\beta,\eta_\gamma\rangle\eta_\beta\eta_\gamma\overline{\psi}-
    \frac{1}{2}\langle \overline{\nabla}_{\eta_\alpha}\eta_\beta,e_j\rangle\eta_\beta e_j\overline{\psi}
    +\frac{1}{4}\langle \overline{\nabla}_{\eta_\alpha}e_i,e_j\rangle e_i e_j\overline{\psi}
    \\=& \eta_\alpha(\overline{\psi})+\left(\sum_{\gamma\neq \alpha}\frac{1}{2}f^\gamma \eta_\alpha\eta_\gamma \overline{\psi}\right)+\frac{1}{2}f^\gamma N^{-1}\eta_\alpha X^\gamma \overline{\psi}+\sum_{\gamma\neq\alpha}\frac{1}{2}f^\gamma N^{-2}X^\alpha X^\gamma \overline{\psi}
    \\=& \eta_\alpha(\overline{\psi})+\frac{1}{2}f^\gamma N^{-1}\eta_\alpha X^\gamma \overline{\psi}+\sum_{\gamma\neq \alpha}f^\gamma \eta_\alpha \eta_\gamma \overline{\psi}.
    \end{split}
\end{equation}
Since $\partial_{\tau_\alpha}\overline{\psi}=0$ and applying Lemma \ref{connection formula}, we have 
\begin{equation}\label{eta2}
    \begin{split}
    \eta_\alpha(\overline{\psi})=&-N^{-1}X^\alpha(\overline{\psi})
    \\=& -\overline{\nabla}_{N^{-1}X^\alpha} \overline{\psi}+\frac{1}{4}\langle \nabla_{N^{-1}X^\alpha}\eta_\beta, \eta_\gamma\rangle \eta_\beta\eta_\gamma \overline{\psi}-\frac{1}{2}\langle \nabla_{N^{-1}X^\alpha}\eta_\beta, e_j\rangle \eta_\beta e_j \overline{\psi}
    \\=& -\frac{1}{2}N^{-1}f^\beta \eta_\beta X^\alpha\overline{\psi}.
    \end{split}
\end{equation}
 Therefore, combining equations \eqref{eta1} and \eqref{eta2} yields $\overline{\nabla}_{\eta_\alpha}\overline{\psi}=0$.  
 
\bibliography{literature.bib}
\bibliographystyle{amsplain}
\end{document}